\documentclass[12pt]{article} \usepackage{fullpage}
\usepackage{lineno}

\usepackage{nicefrac}
\usepackage{url}
\usepackage{color}
\usepackage[american]{babel}
\usepackage{graphicx}
\usepackage{version}
\usepackage{subfigure}
\usepackage[textsize=tiny]{todonotes}

\newif\iffull
\fulltrue
\graphicspath{{figures/}}

\usepackage{amsthm,amsmath,amssymb}
\newtheorem{theorem}{Theorem}
\newtheorem{conjecture}{Conjecture}
\newtheorem{corollary}{Corollary}

\newtheorem{observation}{Observation}

\newtheorem{lemma}{Lemma}

\newcommand{\student}[1]{}
\newcommand{\postdoc}[1]{}

\newcommand{\calC}{{\ensuremath{\cal C}}}

\newcommand{\odd}{{\ensuremath{\it odd}}}

\renewcommand{\subsubsection}[1]{\paragraph{#1}}

\newcommand{\leaveout}[1]{}

\date{}
\title{Matchings in 1-planar graphs with large minimum degree}
\author{Therese Biedl%
\thanks{David R.~Cheriton School of Computer
Science, University of Waterloo, Waterloo, Ontario N2L 1A2, Canada.
Work of TB supported by NSERC.  {\tt biedl@uwaterloo.ca} }
\and John Wittnebel%
\addtocounter{footnote}{-1}\footnotemark
}

\begin{document}

\maketitle
\begin{abstract}
In 1979, Nishizeki and Baybars showed 
that every planar graph with minimum degree 3
has a matching of size $\frac{n}{3}+c$ (where the
constant $c$ depends on the connectivity), and even better
bounds hold for planar graphs with minimum degree 4 and 5.
In this paper, we investigate
similar matching-bounds for {\em 1-planar} graphs, i.e.,
graphs that can be drawn such that every edge has at most one
crossing.  We show that every 1-planar graph with minimum
degree 3 has a matching of size at least $\frac{1}{7}n+\frac{12}{7}$,
and this is tight for some graphs.  We provide similar bounds
for 1-planar graphs with minimum degree 4 and 5, while the
case of minimum degree 6 and 7 remains open.
\end{abstract}

%\linenumbers

%%%%%%%%%%%%%%%%%%%%%%%%%%%%%%%%%%%%%%%%%%%%%%%%%%%%%%%%%%%%%%%%%%%%%%%%
\section{Introduction}

Matchings are one of the oldest and best-studied problems in graph
theory, see for example the extensive reviews of matching theory in
\cite{LP86,Berge73}.
We focus here on matchings in graph classes that are restricted to
have special drawings.  In particular, a graph is called planar
if it can be drawn without crossing in the plane (detailed definitions
are below).  Nishizeki and Baybars \cite{NB79}
argued that every simple planar graph with $n\geq X$ vertices has a matching
of size at least $Yn+Z$, where $X,Y,Z$ depend on the minimum degree $\delta$
and the connectivity $\kappa$ of the graph (they explore all possibilities
of $\delta$ and $\kappa$).  Their bounds are tight in the sense that
some planar graph that satisfies the restrictions has no bigger matching.

The goal of this paper is to develop similar results for simple 1-planar graphs,
i.e., graphs that can be drawn in the plane with at most one crossing per
edge.  These graphs have been of high interest to the graph theory community
ever since Ringel introduced them in 1965 \cite{Ringel1965}; we refer the
reader to an extensive annotated bibliography \cite{KLM17} for
many results.    To our knowledge, no previous matching-bounds were known
for 1-planar graphs of given minimum degree.   We prove here the following:

\begin{theorem}
\label{thm:main}
Any $n$-vertex simple 1-planar graph with minimum degree $\delta$ has a matching $M$ of the following size:
\begin{enumerate}
\vspace*{-2mm}
\itemsep -2pt
\item $|M|\geq \frac{n+12}{7}$ if $\delta=3$ and $n\geq 7$.
\item $|M|\geq \frac{n+4}{3}$ if $\delta=4$ and $n\geq 20$.
\item $|M|\geq \frac{2n+3}{5}$ if $\delta=5$ and $n\geq 21$.
\end{enumerate}
\end{theorem}

All of our bounds are tight in the sense that there are arbitrarily
large simple 1-planar graphs  of the required minimum degree for which no
matching can be larger.  We also provide some simple 1-planar graphs with minimum
degree 6 and 7 with upper bounds on the size of their matchings, though
proving that these can always be achieved remains open.  No 
simple 1-planar graph can have minimum degree 8 or higher.

Our proofs follow similar ideas as the proofs in Nishizeki and Baybars,
but need some new results that do not immediately transfer from planar
to 1-planar graphs.  In particular, at the heart of the proofs in \cite{NB79}
lies the idea that a planar bipartite graph has at most $2n-4$ edges.  
It is known that every 1-planar bipartite graph has at most $3n-6$ edges,
but inserting this into the proof from \cite{NB79} would give
no non-trivial matching-bounds for $\delta\leq 5$.  We therefore need
to develop different techniques to analyze how big an independent set
in a 1-planar graph can be, given bounds on the minimum degree;
this result may be interesting in its own right.

\iffull
\else
\fi

%%%%%%%%%%%%%%%%%%%%%%%%%%%%%%%%%%%%%%%%%%%%%%%%%%%%%%%%%%%%%%%%%%%%%%%%
\subsubsection{Preliminaries}
\label{sec:definitions}

Let $G=(V,E)$ be a graph with $n$ vertices and $m$ edges; to avoid
trivialities assume $n\geq 4$ throughout.    We also assume familiarity
with basic terms in graph theory; see e.g.~\cite{Die12} for details.
A {\em matching} of $G$ is a set of edges for which the endpoints
are all distinct.   An {\em independent set} of $G$ is a set of
vertices without edges between them.
We assume that the input graph $G$ is {\em simple},
i.e., has neither a loop nor a multiple edge.
It will sometimes be convenient to add multiple edges, but only
under restrictions specified below. 
$G$ is {\em connected} if any two vertices are connected via
a path.  
The {\em connectivity} of $G$ is the maximum number $\kappa$ such
that removing any $\kappa-1$ vertices leaves a connected graph.
A {\em component} of $G$ is a maximal connected subgraph; we
call it a {\em singleton} if it has only one vertex.
A {\em bipartite} graph is a graph $G=(V,E)$ where the vertices can
be partitioned into {\em sides} $V=S\cup T$ such that each side
is an independent set.
%A {\em cutting $k$-set} is a set $V'$ of $k$ vertices such that
%$G{\setminus} V'$ has more connected components than $G$; it is 
%called a {\em cutvertex} for $k=1$ and a {\em cutting pair} for
%$k=2$.  $G$ is called {\em $k$-connected} if it has no 
%cutting $(k{-}1)$-set.

Nearly all papers that give lower bounds on matching-sizes
(see e.g.~\cite{NB79,Dill90}) use the Tutte-Berge-Formula
\cite{Berge1958}.

\begin{theorem}[Tutte-Berge]
\label{thm:TB}
The size of a maximum matching $M$ equals the minimum, over
all vertex-sets $S$, of $\frac{1}{2}(n-(\odd(G\setminus S)-|S|))$.  Here,
$\odd(G\setminus S)$ denotes the number of components of odd cardinality in
the graph $G\setminus S$.  
\end{theorem}

%It is easy to see that every odd component must
%either have an unmatched vertex or ``consume'' one vertex of $S$, so
%$\odd(G\setminus S)-|S|$ is an obvious lower bound on the number of unmatched
%vertices in any matching; the non-trivial part of the Tutte-Berge-Formula is
%to argue that for the maximum matching this lower bound is achieved.  
To prove a lower bound on a matching, one uses the following reformulation of the
non-trivial direction of Theorem~\ref{thm:TB}.

\begin{corollary} 
\label{cor:TB}
If $G=(V,E)$ is a graph such that $\odd(G\setminus S)-|S|\leq cn-d$
for all vertex-sets $S\subset V$ and some constants $c,d$, then $G$ has a matching of
size at least $\frac{1-c}{2}n + \frac{d}{2}$.
\end{corollary}

\subsubsection{Planar and 1-planar graphs}
A {\em drawing} $\Gamma$ of a graph $G$ assigns points in $\mathbb{R}^2$ to vertices and curves 
in $\mathbb{R}^2$ to edges.  In what follows, we usually identify the 
element of $G$ (i.e., vertex or edge) with the geometric element in $\Gamma$ (i.e., point or curve) that corresponds
to it.  All drawings are assumed to be {\em good} (see e.g.~\cite{SchaeferBook}
for a detailed discussion), which means among others that no two vertices
coincide, no edge intersects itself or a non-incident vertex, and any two edges
intersect in at most one point and where they either end or fully cross.
(For graphs that are not simple, multi-edges are permitted to
meet twice, once at each end.)  Whenever a drawing is fixed, we think
of the edges incident to a vertex $v$ as cyclically ordered according to
the order in which they end at $v$.

A drawing is called
{\em $k$-planar} if every edge has at most $k$ crossings; in this paper all
drawings are 1-planar and sometimes we restrict the attention to 0-planar
(``planar'') drawings.    
%We also use 
%the term ``drawing'' as a convenient shortcut for ``a graph where one 
%particular drawing has been specified''.
A graph is called {\em planar [1-planar]} if it has a planar [1-planar] drawing.

For the following definitions fix a planar drawing $\Gamma$.
The {\em faces} of $\Gamma$  are the connected pieces of $\mathbb{R}^2{\setminus} \Gamma$,
and described by giving the collection of circuits that form its boundary.
A {\em bigon} is a face whose boundary is a single cycle consisting of two
copies of the same edge.  Our input graph is assumed to be simple,
but we will sometimes add edges for counting arguments, and then allow
multi-edges, but never bigons.  We never allow loops.
Any graph with a planar bigon-free drawing has at most $3n-6$ edges, and
at most $2n-4$ edges if it is bipartite.

For the following definitions fix a 1-planar drawing $\Gamma$.
We call an edge {\em crossed} if it contains a crossing and {\em uncrossed} otherwise.
The {\em planarization} $\Gamma_P$ of $\Gamma$ is the planar
drawing obtained by replacing every crossing with a dummy-vertex of degree 4.  The
{\em regions} of $\Gamma$ are the faces of its planarization $\Gamma_P$,
the {\em corners} of a region of $\Gamma$ are the vertices of the corresponding
face of $\Gamma_P$; corners are vertices or crossings of $\Gamma$.   
%The set of 
%regions is denoted by $\calR$.  A region is called {\em uncrossed} if all its corners
%are vertices, and {\em crossed} otherwise.

%Most terms for planar drawings naturally carry over to a 1-planar drawing $\Gamma$ via
%the planarization $\Gamma_P$, for example 
A {\em bigon} of $\Gamma$
is a bigon of $\Gamma_P$.  Any bigon-free loop-free 1-planar graph has at most $4n-8$ 
edges, and at most $3n-6$ edges if it is bipartite.  We need a slightly more detailed
bound.  
\begin{observation}
\label{obs:edgecount}
Let $G$ be a bipartite graph with a 1-planar bigon-free drawing that has
$m_\times$ crossed and $m_-$ uncrossed edges.  Then $\frac{1}{2}m_\times + m_- \leq 2n-4$.
\end{observation}
\begin{proof}
The crossed edges come in pairs, and if we remove one edge from each pair then
we obtain a planar bipartite bigon-free drawing.  This has at most $2n-4$ edges, and so 
$\frac{1}{2}m_\times + m_- \leq 2n-4$.
\end{proof}

%%%%%%%%%%%%%%%%%%%%%%%%%%%%%%%%%%%%%%%%%%%%%%%%%%%%%%
\section{1-planar graphs without large matchings}

In this section, we create some 1-planar graphs that have large minimum
degree and for which the maximum matching is small.  
%We will see in the next section that the obtained bounds are tight.

\begin{figure}[ht]
\hspace*{\fill}
\subfigure[~]{\includegraphics[scale=0.7,page=1,trim=0 0 0 0,clip]{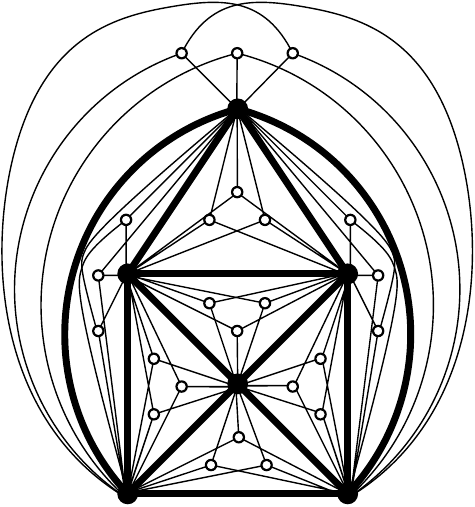}}
\hspace*{\fill}
\subfigure[~]{\includegraphics[scale=0.7,page=3,trim=0 0 0 0,clip]{tight.pdf}}
\hspace*{\fill}
\subfigure[~]{\includegraphics[scale=0.9,page=4,trim=20 0 40 0,clip]{tight.pdf}}
\hspace*{\fill}
\subfigure[~]{\includegraphics[scale=0.9,page=5,trim=10 0 10 0,clip]{tight.pdf}}
\hspace*{\fill}
\caption{Graphs that do not have large matchings.  Graph $H$
is bold, vertices in $S$ are black.  (a) Minimum degree 3.  (b and c) Minimum degree 4.  (d) Minimum degree 5.}
\label{fig:tight}
\end{figure}

\begin{lemma}
\label{lem:mindeg3}
For any $N$,
there exists a simple 1-planar graph with minimum degree 3 and $n\geq N$ vertices for which any matching
has size at most $\frac{n+12}{7}$.
\end{lemma}
\begin{proof}
Consider the graph in Fig.~\ref{fig:tight}(a), which has been built as follows.
Start with an arbitrary planar graph $H$ on $s$ vertices, where $s\geq \max\{3,\frac{N+12}{7}\}$, such that all faces of $H$ are triangles.
Into each face $\{u,v,w\}$ of $H$, insert three more vertices that are
all adjacent to all of $u,v,w$.  Obviously the resulting graph $G$ has
minimum degree 3, and the figure shows that $G$ is 1-planar.
Also, $H$ has $2s-4$ faces, hence $G$ has $n=s+3(2s-4)=7s-12\geq N$ vertices.
Setting $S$ to be the $s$ vertices of $H$, we observe that every vertex
in $G\setminus S$ becomes a singleton component.  So 
$$\odd(G\setminus S)-|S| = 3(2s-4) - s = 5s-12 = \frac{5n-24}{7}.$$
By Theorem~\ref{thm:TB} therefore any matching has size at most $\frac{n+12}{7}$.
% example: s=6, n=30, odd(G)=18, M<=6
\end{proof}

\begin{lemma}
\label{lem:mindeg4}
For any $N$,
there exists a simple 1-planar graph with minimum degree 4 and $n\geq N$ vertices for which any matching
has size at most $\frac{n+4}{3}$.
\end{lemma}
\begin{proof}
Consider the graph in Fig.~\ref{fig:tight}(b), which has been built as follows.
Start with a planar graph $H$ on $s$ vertices, where $s\geq \max\{4,\frac{N+4}{3}\}$,
such that all faces of $H$ are simple cycles of length 4.
Into each face $\{u,v,w,x\}$ of $H$, insert two more vertices that are
all adjacent to all of $u,v,w,x$.  Obviously the resulting graph $G$ has
minimum degree 4 and the figure shows that $G$ is 1-planar.
Also, $H$ has $s-2$ faces, hence $G$ has $n=s+2(s-2)=3s-4\geq N$ vertices.
Setting $S$ to be the $s$ vertices of $H$, we observe that every vertex
in $G\setminus S$ becomes a singleton component.  So 
$$\odd(G\setminus S)-|S| = 2(s-2) - s = s-4 = \frac{n-8}{3}.$$
By Theorem~\ref{thm:TB} therefore any matching has size at most $\frac{n+4}{3}$.
\end{proof}

We note here that the same lower bound can be achieved with a much simpler construction
that combines $(n-2)/3$ copies of the complete graph $K_5$ at an edge
(see Fig.~\ref{fig:tight}(c)), but the connectivity of the resulting graph is not as high.

\begin{lemma}
\label{lem:mindeg5}
For any $N$,
there exists a simple 1-planar graph with minimum degree 5 and $n\geq N$ vertices for which any matching
has size at most $\frac{2n+3}{5}$.
\end{lemma}
\begin{proof}
Consider the graph in Fig.~\ref{fig:tight}(d), which has been built as follows.
Let $n\geq N$ be such that $n\equiv 1\bmod 5$.  Create one vertex $v_s$, and split the
remaining vertices into $(n-1)/5$ groups of five vertices each.  For each group
$\{v_1,\dots,v_5\}$, inserted edges to turn 
$\{s,v_1,\dots,v_5\}$ into a complete graph $K_6$.  
Obviously the resulting graph $G$ has
minimum degree 5, and the figure shows that $G$ is 1-planar.

Setting $S$ to be the single vertex $v_s$, we observe that each of the
$(n-1)/5$ groups become an odd component of $G\setminus S$. Hence
$$\odd(G\setminus S)-|S| = (n-1)/5 - 1 = \frac{n-6}{5}.$$
By Theorem~\ref{thm:TB} therefore any matching has size at most $\frac{2n+3}{5}$.
\end{proof}

%%%%%%%%%%%%%%%%%%%%%%%%%%%%%%%%%%%%%%%%%%%%%%%%%%%%%%
\section{Lower bounds on the matching-size}

In this section, we prove Theorem~\ref{thm:main} by proving bounds
on $\odd(G\setminus S)-|S|$ for a 1-planar graph $G$ and an arbitrary
vertex-set $S$ under assumptions on the minimum degree.

We first briefly review the technique by Nishizeki and Baybars \cite{NB79}
to prove matchings bounds in a planar graph $G$ of minimum degree 3.
Let a vertex-set $S$ be given.
By Corollary~\ref{cor:TB} it suffices to show that $\odd(G\setminus S)-|S|\leq \frac{n+8}{3}$.  To this end,
delete any edge that
connects two vertices in $S$, and delete any component of $G\setminus S$
that has even cardinality;
note that both do not increase $\odd(G\setminus S)$ and can only decrease $n$.
Also delete any odd component of $G\setminus S$
that contains at least 3 vertices; this decreases
$\odd(G\setminus S)$ by one and $n$ by at least 3 and so does not make the
bound worse.  We end with a planar bipartite graph $G'$
where one side is $S$ and
the other side $T$ has one vertex for each singleton component of $G\setminus S$.
Furthermore, no edges incident to a vertex in $T$ was deleted, so $\deg(t)\geq 3$ for
all $t\in T$.  Since $G'$ has $n(G')=|S|+|T|$ vertices, it has at most 
$2(|S|{+}|T|)-4$ edges and at least $3|T|$
edges, so $|T|\leq 2|S|-4$.  Therefore $3\left(\odd(G'\setminus S)-|S|\right)=3|T|-3|S|\leq 2|T|-4|S|+n(G')\leq n(G')-8$
whence the matching-bound follows.

\subsection{Independent sets in 1-planar graphs}

The crucial ingredient for the proof by Nishizeki and Baybars is the
bound $|T|\leq 2|S|-4$ in a planar bipartite graph $(S\cup T,E)$
where all vertices in $T$ have minimum degree 3.  In this section, we
aim to show similar bounds for 1-planar graphs.  This requires
entirely different techniques than the simple edge-counting argument
that sufficed for planar graphs.  We phrase it as a slightly more
general statement about independent sets in 1-planar graphs.

\begin{lemma}
\label{lem:bipartite}
Let $G$ be simple 1-planar graph. 
Let $T$ be a non-empty independent set in $G$
where $\deg(t)\geq 3$ for all $t\in T$.
Let $T_d$ be the vertices in $T$ that have degree $d$. 
Then 
$$2|T_3|+\sum_{d\geq 4} (3d-6)|T_d| \leq 12|V\setminus T|-24.
\footnote{This bound is not tight, and in fact
$2|T_3|+\sum_{d\geq 4} (3d+3\lceil d/3 \rceil -12) |T_d| \leq 12|V\setminus T|-24$
could be shown with much the same proof.}$$
\end{lemma}

We will at the same time prove another result that is related
(but neither result implies the
other); this will be used in a future paper \cite{BiedlKlute20}.  
Define the {\em crossing-weighted degree} of a vertex $v\in V$ to be
the degree plus the number of incident uncrossed
edges.  Thus, uncrossed edges count doubly.    

\begin{lemma}
\label{lem:bipartiteWeighted}
Let $G$ be a simple graph with a 1-planar drawing $\Gamma$.
Let $T$ be a non-empty independent set in $G$
where $\deg(t)\geq 3$ for all $t\in T$.
Let $W_d$ be the vertices in $T$ that have crossing-weighted degree $d$.
Then $$2|W_3|+2|W_4|+\sum_{d\geq 5} (3d-12) |W_d| \leq 12|V\setminus T|-24.$$
\end{lemma}

\begin{proof} (of both Lemma~\ref{lem:bipartite} and \ref{lem:bipartiteWeighted})
Fix a 1-planar drawing $\Gamma$ of $G$ if not given yet.
We use a charging scheme, where we assign some {\em charges}
(units of weight) to edges in $G$ (as well as some additions
that we make to $G$), redistribute those to the vertices in $T$,
and then count the number of charges in two ways to obtain the 
bound.

\medskip\noindent{\bf Step 0: Make $G$ bipartite.}  Let $S$ be the
vertices that are not in $T$, and note that $|S|\geq 3$ since $T$ is
non-empty and vertices in it have degree 3 or more.
Delete all edges within $S$ so that
$G$ becomes bipartite; this does not affect degrees in $T$, can only 
increase crossing-weighted degrees in $T$ by making some edges uncrossed,
and so it suffices to prove the bound for the resulting graph.

\medskip\noindent{\bf Step 1: Add more edges.}
Now add any edge to $\Gamma$ that can be added without crossing while
remaining bipartite.  This can only increase degrees of vertices
in $T$ and so improve the bound.  We are allowed to add multiple edges, as
long as they do not form a bigon, 
see edge $({\tt A},{\tt d})$ in Fig.~\ref{fig:example}.

We claim that in $\Gamma'$ no vertex $t\in T$ has three
consecutive crossed edges $e_1,e_2,e_3$ (in the cyclic order
of edges defined by $\Gamma'$).
Assume there were three such consecutive edges in $\Gamma$, 
and let $(t',s')$ (with $t'\in T$ and $s\in S$)
be the edge that crosses $e_2$, say at $c$.
(For an illustration, consider vertex $t={\tt d}$ in Fig.~\ref{fig:example},
which has three consecutive crossed edges to ${\tt E},{\tt A}$ and ${\tt C}$ in $\Gamma$; 
hence $(t',s')=({\tt c},{\tt B})$.)
We can add an uncrossed edge $e'=(s',t)$ by 
tracing along $c$;  this  edge would end
before or after $e_2$ in the clockwise order at $t$.    Adding $e'$
does not create a bigon, since on one side of $e'$ the region contains
$c$ and on the other side the region contains the
crossing in $e_1$ or $e_3$.
So we added this edge when creating $\Gamma'$ from $\Gamma$, and no
three consecutive crossed edges in $\Gamma$ remain consecutive in $\Gamma'$.

%For any $t\in T$, let $d(t)$ be
%the degree of $t$ in $\Gamma'$, and note that $d(t)\geq \deg(t)\geq 3$
%(where $\deg(t)$ is the degree in $G$).  
%Also let $d_-(t)$
%be the number of uncrossed edges incident to $t$ in $\Gamma'$.    

\begin{figure}[ht]
\hspace*{\fill}
\includegraphics[width=0.6\linewidth]{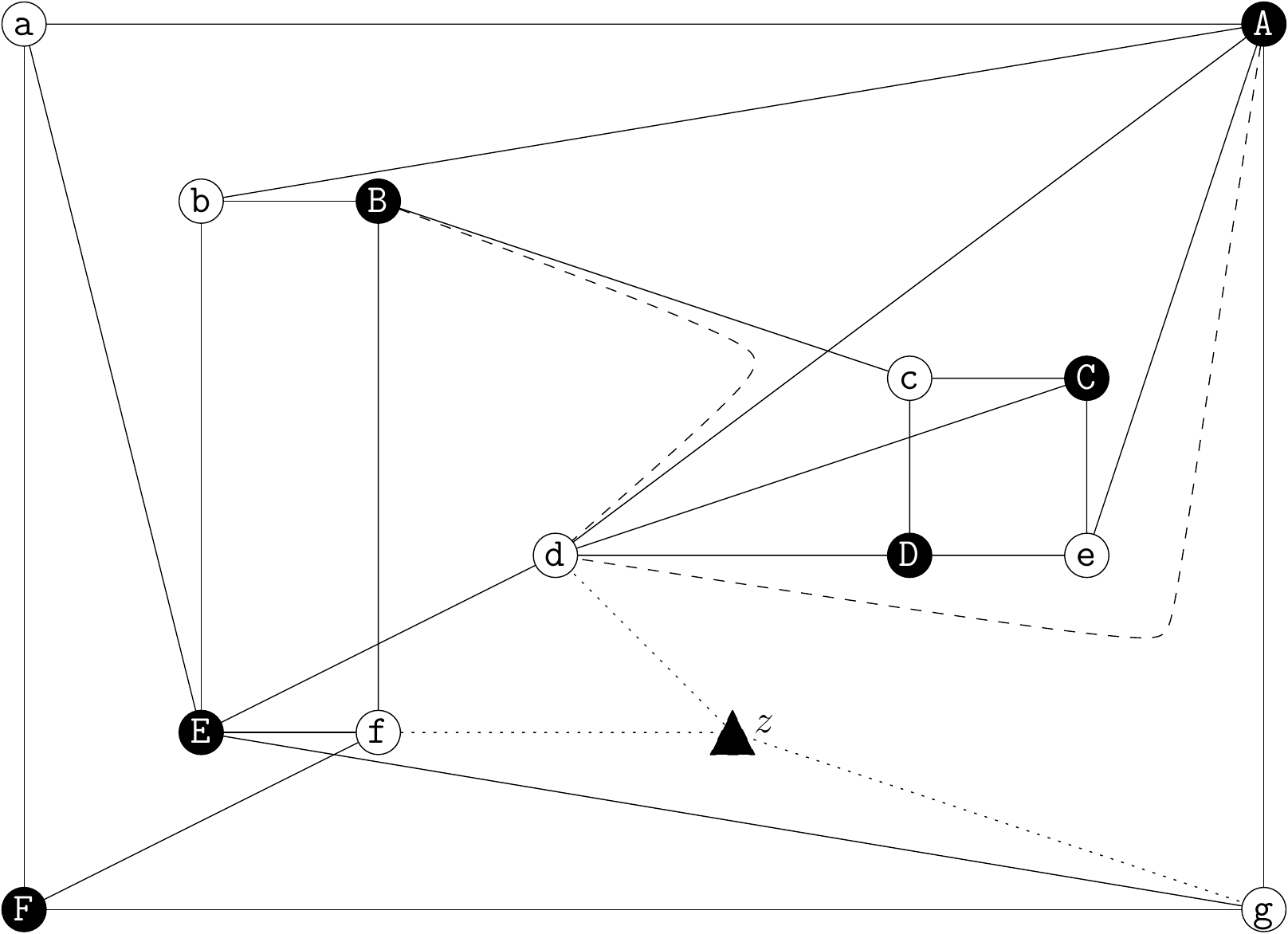}
\hspace*{\fill}
\caption{A drawing $\Gamma$ of a bipartite 1-planar graph with minimum degree 3;
vertices of $T$ are white.
Edges in $\Gamma'\setminus \Gamma$ are dashed, the unique vertex in $S_\Delta$ is
a triangle and edges in $E_\Delta$ are dotted.}
\label{fig:example}
\end{figure}

\medskip\noindent{\bf Step 2: Add vertices and edges.}
As a next step, we possibly augment $\Gamma'$ further with vertices $S_\Delta$.
Assume there exists
a region $R$ in $\Gamma'$ incident to at least three vertices in $T$.
Add a new vertex $z$ to $S_\Delta$ and connect it to exactly three
vertices of $T$ on $R$ via uncrossed edges within $R$.  
This splits $R$ into three regions, each with
fewer vertices of $T$, so we can repeat until no such regions remain.
See the new vertex incident to ${\tt d,f,g}$ in
Fig.~\ref{fig:example}.  Let $E_\Delta$ be the added
edges; we have $|E_\Delta|=3|S_\Delta|$.  Also observe that edges in
$E_\Delta$ are uncrossed and that the resulting
drawing $\Gamma_\Delta$ is again bipartite.

\medskip\noindent{\bf Step 3: Assigning charges.}
We assign charges as follows:  Let $E_-$ be the uncrossed edges of $\Gamma'$;
each of those receives 6 charges.  Let $E_\times$ be the crossed edges
of $\Gamma'$; each of those receives 3 charges.  Finally the (uncrossed) edges $E_\Delta$ 
of $\Gamma_\Delta-\Gamma'$ receive 2 charges each.
Using Observation~\ref{obs:edgecount} hence
%Since $E_-\cup E_\Delta$ are the uncrossed edges of $\Gamma_\Delta$ and
%$E_\times$ are the crossed edges of $\Gamma_\Delta$, we have $\frac{1}{2} |E_\times| + |E_-|+|E_\Delta|\leq 2(|S|+|S_\Delta|+|T|)-4$
%and hence
\begin{eqnarray}
\label{eq:1}
\nonumber \#\mbox{charges} & = & 6|E_-|+3|E_\times| + 2|E_\Delta|
= 6(|E_-|+|E_\Delta|)+3|E_\times| - 4|E_\Delta| \\
\nonumber & \leq & 12|S|+12|S_\Delta|+12|T| -  24 - 12|S_\Delta| \\ 
& = & 12|S|+12|T|-24.
\end{eqnarray}

\medskip\noindent{\bf Step 4: Charges at a vertex.}
For $t\in T$, let $c(t)$ be the total charges of incident edges of $t$.
We lower-bound $c(t)$ as follows:
\begin{itemize}
\item Assume first that $t$ has at least two incident uncrossed edges 
in $\Gamma'$.
It obtains 12 charges from these two edges, and
at least $3(\deg(t)-2)$ further charges from the remaining edges
that it had in $G$.
Hence $c(t)\geq 12+3(\deg(t)-2) =3\deg(t)+6 \geq 15$.

\item Now assume that $t$ has at most one uncrossed incident edge in $\Gamma'$.
We aim to show that $c(t)\geq 14$.
By $\deg(t)\geq 3$, and since no three crossed edges are consecutive, vertex $t$ has at least one 
uncrossed edge, so it has exactly one, call it $e_1$.  This implies that $t$ has only three
incident edges in $\Gamma'$, else the edges other than $e_1$ would
contain three consecutive crossed edges.  So this case can occur only 
if $t\in T_3$ and it has exactly one uncrossed edge $e_1=:(t,s_1)$ and two
crossed edges $(t,s_2)$ and $(t,s_3)$.
See vertex $t={\tt f}$ in Fig.~\ref{fig:example}, where $\{s_2,s_3\}=\{{\tt B},{\tt F}\}$.

Let $(s',t')$ be the edge that crosses $(t,s_2)$ in $\Gamma$, with $s'\in S$
and $t'\in T$.  Then (as above) an edge $(s',t)$ could be drawn without crossing
and could have been added to $\Gamma'$. Since the only uncrossed edge at $t$
is $e_1$ therefore $s'=s_1$.
Similarly one argues that $(t,s_3)$ is crossed by edge $(s_1,t'')$
for some $t''\in T$.  (In Fig.~\ref{fig:example} we have $\{t',t''\}=\{{\tt d,g}\}$.)
If we had $t'=t''$ then there would be two copies of edge $(s_1,t')$, and
both would be crossed.  Since $G$ is simple and no crossed edges were added
for $\Gamma'$, this is impossible and $t'\neq t''$.      

Observe that hence $t,t',t''$ all belong to the same region $R$ between the
two crossings in $(t,s_2)$ and $(t,s_3)$.  Therefore we added a vertex $z$ of $S_\Delta$ inside 
$R$ and made it adjacent to $t$.  Edge $(t,z)$ has two charges, and in
total we have $c(t)\geq 6+3+3+2=14$. 
\item (The following is relevant only for Lemma~\ref{lem:bipartiteWeighted}.)
	Assume $t\in W^d$.  Then $t$ receives 6 charges for every uncrossed
	edge that was incident in $\Gamma$, and 3 charges for every crossed 
	edge that was incident in $\Gamma$ (and possibly some more from
	edges added in later steps).  Since uncrossed edges count twice
	for the crossing-weighted degree, hence $c(v)\geq 3d$.
\end{itemize}

To prove Lemma~\ref{lem:bipartite}, use
$c(t)\geq 14$ for all $t$, $c(t)\geq 18$ for $t\in T_4$
and $c(t)\geq 21$ for $t\in T_d$ with $d\geq 5$ since a vertex in $T_d$ has degree at least $d$ in $\Gamma'$.  Therefore
\begin{eqnarray}
\label{eq:2}
\#\mbox{charges} & = & \sum_{t\in T} c(t) \geq 14|T_3| + 18|T_4| + 21 \sum_{d\geq 5} |T_d|.
\end{eqnarray}
To prove Lemma~\ref{lem:bipartiteWeighted}, observe that again
$c(t)\geq 14$ for all $t$ and $c(t)\geq 3d$ for $t\in W_d$.  Therefore
\begin{eqnarray}
\label{eq:3}
\#\mbox{charges} & = & \sum_{t\in T} c(t) \geq 14|W_3| + 14|W_4| + \sum_{d\geq 5} 3d\; |W_d|.
\end{eqnarray}
Combining this with (\ref{eq:1}) and subtracting $12|T|=12\sum_{d\geq 3} |T_d|
=12\sum_{d\geq 3} |W_d|$
from both sides gives the results.
\end{proof}

\subsection{Matching-bounds}

Now we use Lemma~\ref{lem:bipartite} to obtain the desired matching-bounds.
For minimum degree 3 and 4 we proceed almost exactly as done by Nishizeki and
Baybars \cite{NB79}: preprocess the graph to remove some edges and components that can do
no harm, and then use the upper bound on the resulting independent set in a 1-planar graph.

\begin{lemma}
\label{lem:mindeg34}
Let $G$ be a simple 1-planar graph with minimum degree $\delta\geq 3$.
Then for any vertex set $S$ with $|S|\geq 2$, we have 
\begin{itemize}
\item $\odd(G\setminus S)-|S|\leq \frac{5}{7}n-\frac{24}{7}$ if $\delta\geq 3$, and
\item $\odd(G\setminus S)-|S|\leq \frac{1}{3}n-\frac{8}{3}$ if $\delta\geq 4$.
\end{itemize}
\end{lemma}
\begin{proof}
Set $c_3=\frac{5}{7}$, $d_3=\frac{24}{7}$, $c_4=\frac{1}{3}$ and $d_4=\frac{8}{3}$;
the goal is to show $\odd(G\setminus S)-|S|\leq c_\delta n - d_\delta$.

We first preprocess $G$ by removing any component of $G\setminus S$
that has even size.  This does not affect $|S|$ 
or $\odd(G\setminus S)$, or degrees of vertices in $V\setminus S$ that remain,
and it can only decrease $n$.  So it suffices to prove the
bound for the remaining graph.

Next remove all odd components of
$G\setminus S$ that have three or more vertices,
and let $G'$ be the resulting graph.
If $k$ components are removed, then hence 
$n(G')\leq n-3k$.  We will show below that $\odd(G'\setminus S)\leq c_\delta n(G') - d_\delta$,
and therefore (by $c_\delta \geq \frac{1}{3}$)
$$ \odd(G\setminus S) = \odd(G'\setminus S)+ k \leq c_\delta n(G')- d_\delta + c_\delta 3k
\leq c_\delta n - d_\delta.$$
It remains to show the claim for $G'$.  Let $T=V(G')\setminus S$, and notice
that these are exactly the singleton components of $G\setminus S$ since all other components were removed.
In particular they form an independent set in $G'$.
Let $T_d$ be the vertices in $T$ that have degree $d$ in $G'$.

If $|T|=0$  then $G'\setminus S$ is empty, so $n(G')=|S|\geq 2$ and
$$\odd(G'\setminus S)-|S| = -|S|\leq -2 = 2 c_\delta - d_\delta \leq c_\delta n - d_\delta$$
as desired.  If $T$ is non-empty, then
apply Lemma~\ref{lem:bipartite} to $G'$, and also observe that $n(G')=|S|+|T|$.
If $\delta=3$, then $2|T| = \sum_{d\geq 3} 2|T_d|\leq 2|T_3|+6|T_4|+9\sum_{d\geq 5} |T_d| \leq 12|S|-24$ by Lemma~\ref{lem:bipartite}. Therefore
$$7\odd(G'\setminus S)-7|S| = 7|T|-7|S| = 2|T|-12|S|+5n(G') \leq 5n(G')-24.$$
If $\delta=4$, then $T_3$ is empty and $6|T| = \sum_{d\geq 4} 6|T_d|\leq 12|S|-24$ by Lemma~\ref{lem:bipartite}. Therefore
$$9\odd(G'\setminus S)-9|S| = 9|T|-9|S| = 6|T|-12|S|+3n(G') \leq 3n(G')-24.$$
The desired bound follows by dividing suitably.
\end{proof}	

With this we can obtain the first two matching bounds.

\begin{proof} (of Theorem~\ref{thm:main}(a) and (b))
Let $G$ be a 1-planar graph with minimum degree $\delta\in \{3,4\}$.
Fix an arbitrary vertex set $S$.  If $|S|\geq 2$ then 
Lemma~\ref{lem:mindeg34} gives $\odd(G\setminus S)-|S|\leq \frac{5}{7}n-\frac{24}{7}$
and $\odd(G\setminus S)-|S|\leq \frac{1}{3}n-\frac{8}{3}$, respectively,
and Corollary~\ref{cor:TB} gives the result. So we only have to bound
$\odd(G\setminus S)-|S|$ for small $S$.  
Let $X$ be the smallest odd integer with $X\geq \delta+1-|S|$, and
note that any odd component of $G\setminus S$ must have at least $X$ vertices
by simplicity and the minimum degree requirement.   
So $\odd(G\setminus S)\leq \frac{n-|S|}{X}$.  Now distinguish
cases.
\begin{itemize}
\item If $|S|=0$ then $X=5$ and $\odd(G\setminus S)-|S| \leq \frac{n}{5}$.
This is 
at most $\frac{5}{7}n-\frac{24}{7}$ for $n\geq 7$ and
at most $\frac{1}{3}n-\frac{8}{3}$ for $n\geq 20$. 
\item If $|S|=1$ and $\delta=3$ then $X=3$ and $\odd(G\setminus S)-|S| \leq \frac{n-1}{3}-1 \leq \frac{5}{7}n-\frac{24}{7}$ by $n\geq 6$.
\item If $|S|=1$ and $\delta=4$ then $X=5$ and $\odd(G\setminus S)-|S| \leq \frac{n-1}{5}-1 \leq \frac{1}{3}n-\frac{8}{3}$ by $n\geq 11$.
%\item If $|S|=2$ then $X=3$ and $\odd(G\setminus S)-|S| \leq \frac{n-2}{3}-2 = \frac{1}{3}n-\frac{8}{3} \leq \frac{5}{7}n-\frac{24}{7}$
%		by $n\geq 2$.
\end{itemize}
\end{proof}

For graphs with minimum degree 5 we have to keep more odd components of $G$.

\begin{lemma}
\label{lem:oddmindeg5}
Let $G$ be a 1-planar simple graph with minimum degree 5.
Then for any vertex set $S$ with $|S|\geq 1$, we have 
$\odd(G\setminus S)-|S|\leq \frac{1}{5}n-\frac{6}{5}$.
\end{lemma}
\begin{proof}
As in the proof of Lemma~\ref{lem:mindeg34} we can remove
components of $G\setminus S$ that have even size without affecting the bound.
Let $G'$ be the graph obtained from $G$ by removing all odd components of
$G\setminus S$ that have five or more vertices.  If $k$ components are
removed, then hence 
$n(G')\leq n-5k$.  We will show below that $\odd(G'\setminus S)\leq \frac{1}{5} n(G') - \frac{6}{5}$,
and therefore 
$$ \odd(G\setminus S) = \odd(G'\setminus S)+ k \leq \frac{1}{5} n(G')- \frac{6}{5} + \frac{1}{5} 5k
\leq \frac{1}{5} n - \frac{6}{5}.$$

It remains to show the claim for $G'$.  If $\odd(G'\setminus S)=0$, then by $n(G')=|S|\geq 1$
and $\odd(G'\setminus S)-|S|\leq -1 = \frac{1}{5}-\frac{6}{5}\leq \frac{1}{5}n(G')-\frac{6}{5}$ and we are done.
So assume $\odd(G'\setminus S)>0$.
In contrast to the proof of Lemma~\ref{lem:mindeg34}, $G'\setminus S$ is not necessarily an independent set, because components of size 3 
may have edges within them.  In contrast to the approach taken by 
Nishizeki and Baybars \cite{NB79}, we cannot contract such components
into a vertex, because 1-planarity may not be preserved under
contraction.  So we need a different approach.

For $i=1,3$, let $\calC_i$ be the components of $G'\setminus S$ that have size $i$.
Use $V(\calC_i)$ for the vertices of components in $\calC_i$, hence $|V(\calC_1)|=|\calC_1|$
while $|V(\calC_3)|=3|\calC_3|$.
Let $H$ be the 1-planar graph obtained by deleting all edges
within $V(\calC_3)$; this makes $V(\calC_3)\cup V(\calC_1)$ an independent set in $H$.
Any vertex $v \in V(\calC_3)$
has at least five neighbours in $G$ and at most two neighbours in its odd
component, so $\deg_H(v)\geq 3$.  All vertices in $V(\calC_1)$
have degree at least 5 in $H$.
Let $T_d$ for $d\geq 3$ be the vertices of
degree $d$ in $V(\calC_3)\cup V(\calC_1)$, then any vertex in $T_3\cup T_4$
must belong to $V(\calC_3)$.
Applying Lemma~\ref{lem:bipartite} to $H$, therefore
$$12|S|-24 \geq 2|T_3|+6|T_4| + 9 \sum_{d\geq 5} |T_d|
\geq 2|V(\calC_3)| + 9 |V(\calC_1)|.$$
Since $n(G')=|V(\calC_1)|+|V(\calC_3)|+|S|$, this implies
\begin{eqnarray*}
21\left(\odd(G'\setminus S) - |S|\right) & = & 21|\calC_1| + 21 |\calC_3| - 21|S| \\
& = & 21 |V(\calC_1)| + 7 |V(\calC_3)| - 21 |S|  \\
& = & 18 |V(\calC_1)| + 4 |V(\calC_3)| - 24 |S|  + 3n(G') \\
& \leq & (24|S|-48) - 24|S|+3n(G') = 3n(G')-48 
\end{eqnarray*}
which gives $\odd(G'\setminus S)-|S| \leq \frac{1}{7}n(G')-\frac{16}{7} \leq \frac{1}{5}n(G')-\frac{6}{5}$.
\end{proof}

%Note that the bottleneck in the above proof occurs if all odd components have
%size exactly 5 and $|S|=1$.  

\begin{proof} (of Theorem~\ref{thm:main}(c))
Let $G$ be a 1-planar graph with minimum degree 5.
Fix an arbitrary vertex set $S$.  If $|S|\geq 1$ then the
previous lemma gives $\odd(G\setminus S)-|S|\leq \frac{1}{5}n-\frac{6}{5}$.
If $|S|=0$, then every odd component of $G\setminus S=G$ must contain
at least 7 vertices since the minimum degree is 5; 
hence $\odd(G\setminus S)-|S| \leq \frac{n}{7} \leq \frac{1}{5}n-\frac{6}{5}$ by $n\geq 21$.
Either way the bound follows from Corollary~\ref{cor:TB}.
\end{proof}

\section{Other classes of 1-planar graphs}

In this section, we construct some other 1-planar graphs that
do not have large matchings, and offer some conjectures.

\subsection{Non-simple 1-planar graphs}

Our matching bounds were proved for simple 1-planar graphs.  Obviously
no non-trivial matching bounds can exist if we permit bigons, because 
then $K_{2,n}$ (with edges repeated as needed to achieve any desired minimum
degree) has no matching with more than 2 edges.

In fact, even excluding bigons is not enough to ensure a matching of
linear size.  The 1-planar drawing in Fig.~\ref{fig:conjecture}(a) has no 
bigon and the graph has minimum degree 3, yet it has no matching of size
exceeding 2 since removing the two black vertices leaves behind $n-2$
singleton components.  

The story is different for minimum degree $\delta\geq 4$ or if 
no two copies of an edge are both crossed.  Inspecting the
proof of Lemma~\ref{lem:bipartite}, we see that simplicity is
used in two places: we need that no bigons exist to use
Observation~\ref{obs:edgecount}, and we need to exclude the
possibility that $t'=t''$ in the second case of Step 4, where we
bound $c(t)\geq 14$ if $t$ has only one incident uncrossed edge.  This case is
never needed if $\deg(t)\geq 4$, and only uses that no two crossed edges
$(s_1,t')$ exist otherwise.  So Theorem~\ref{thm:main} holds 
for any 1-planar graph with a bigon-free drawing and additionally
the minimum degree is at least 4 and/or there are no multiple crossed 
copies of an edge.

\subsection{1-planar graphs of higher minimum degree}

For planar graphs, matching-bounds are interesting only for $\delta=3,4,5$, because for
$\delta=2$ there are no linear bounds (consider $K_{2,n}$), 
and for $\delta\geq 6$ there exists no planar bigon-free graph of minimum degree 
$\delta$.  In contrast to this, there are simple 1-planar graphs of minimum degree 6 or 7,
while no simple 1-planar graph can have minimum degree $\delta\geq 8$ since it has at most $4n-8$ edges
\cite{Schumacher}.  Naturally
one wonders what matching bounds can be obtained for $\delta=6,7$.

\begin{lemma}
For any $N$,
there exists a simple 1-planar graph with minimum degree 6 and $n\geq N$ vertices for which any matching
has size at most $\frac{3}{7}n+\frac{4}{7}$.
\end{lemma}
\begin{proof}
Consider the graph in Fig.~\ref{fig:conjecture}(b), which has been built as follows.
Let $n\geq N$ be such that $n\equiv 1\bmod 7$.  Create one vertex $v_s$, and split the
remaining vertices into $(n-1)/7$ groups of seven vertices each.  For each group
$\{v_1,\dots,v_7\}$, inserted edges to turn 
$\{v_s,v_1,\dots,v_7\}$ into a cube plus a crossing within each face of the cube.
Obviously the resulting graph $G$ has
minimum degree 6, and the figure shows that $G$ is 1-planar.

Setting $S$ to be the single vertex $v_s$, we observe that each of the
$(n-1)/7$ groups become an odd component of $G\setminus S$. Hence
$$\odd(G\setminus S)-|S| = (n-1)/7 - 1 = \frac{n-8}{7}.$$
Therefore any matching has size at most $\frac{3n+4}{7}$.
\end{proof}

\begin{figure}[ht]
\hspace*{\fill}
\subfigure[~]{\includegraphics[scale=0.8,page=7]{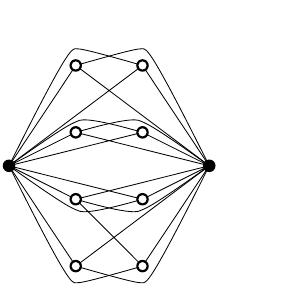}}
\hspace*{\fill}
\subfigure[~]{\includegraphics[scale=0.8,page=5]{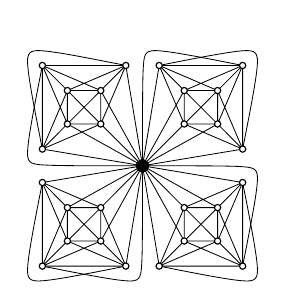}}
\hspace*{\fill}
\subfigure[~]{\includegraphics[scale=0.8,page=2]{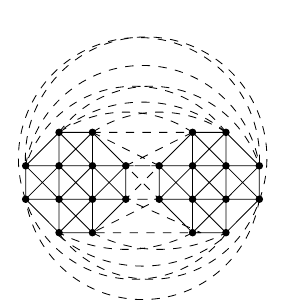}}
%\hspace*{\fill}
%\subfigure[~]{\includegraphics[scale=0.9,trim=0 20 30 0,clip,page=6]{conjecture.pdf}}
%\hspace*{\fill}
%\subfigure[~]{\includegraphics[scale=0.7,page=3]{conjecture.pdf}}
\hspace*{\fill}
\caption{Other 1-planar graphs without large matchings.  (a) 
A bigon-free 1-planar graph with minimum degree 3 and maximum
matching size 2.  (b)
A 1-planar graph with minimum degree 6. 
(c) A 1-planar graph with minimum degree 7.  }
%(d) Building a 3-connected 1-planar graph with minimum degree 5.}
\label{fig:conjecture}
\label{fig:mindeg7}
\end{figure}

We suspect that this is tight.

\begin{conjecture}
Any 1-planar graph with minimum degree 6 and $n\geq N$ vertices has
a matching of size at least $\frac{3}{7}n+O(1)$.
\end{conjecture}

One may wonder why the proof of Lemma~\ref{lem:oddmindeg5} cannot be
generalized to minimum degree 6.    The problem are components of
size 5 in $G\setminus S$.  If we remove these to obtain $G'$, then
we remove $5k$ vertices for $k$ odd components and cannot hope for
an upper bound better than $\frac{n}{5}+O(1)$ for $\odd(G\setminus S)-|S|$.
If we keep components of size 5 in $G'$, then each vertex $t$ of a component $C$ of
size 5 could have four neighbours in $C$, hence only two neighbours
in $S$, and Lemma~\ref{lem:bipartite} cannot be used. 

For minimum degree 7, we similarly have a graph without a perfect matching,
but only conjectures as to whether this is tight.

\begin{lemma}
\label{lem:mindeg7}
For any $N$,
there exists a simple 1-planar graph with minimum degree 7 and $n\geq N$ vertices for which any matching
has size at most $\frac{11}{23}n+\frac{12}{23}$.
\end{lemma}
\begin{proof}
Let $n\geq N$ be such that $n\equiv 1\bmod 23$.  Create one vertex $v_s$, and split the
remaining vertices into $(n-1)/23$ groups of 23 vertices each.  For each group, insert
edges to turn these 23 vertices, plus vertex $v_s$, into one a simple 1-planar graph 
of minimum degree 7, see Fig.~\ref{fig:conjecture}(c).  
Setting $S$ to be the single vertex $v_s$, we observe that each of the
$(n-1)/23$ groups become an odd component of $G\setminus S$. Hence
$$\odd(G\setminus S)-|S| = \frac{n-1}{23} - 1 = \frac{n-24}{23}.$$
Therefore any matching has size at most $\frac{11n+12}{23}$.
\end{proof}

The above example is best in the sense that any 1-planar simple graph with
minimum degree 7 has at least 24 vertices \cite{Bie19-mindeg7}.  However, exploiting
this to show that the bound in Lemma~\ref{lem:mindeg7} is tight remains an
open problem.  

\section{Open problems}
\label{sec:conclusion}

We leave some gaps in our analysis; in particular one wonders what
matching-bounds are tight for simple 1-planar graphs of
minimum degree 6 or 7.
Also, there are many other related graph classes that could be studied.
What, for example, is the size of matchings in 2-planar graphs of
minimum degree $\delta$?

\iffalse
Finally, our bounds rely on the Tutte-Berge-theorem, and as such, do
not give rise to algorithms to find matchings of the proved size, other
than using a general-purpose maximum matching algorithm such as for
example \cite{MV80,Vaz94-matching}.  For planar graphs with minimum
degree 3 (which includes maximum planar graphs), algorithms 
have been designed to find large matchings 
in linear time \cite{FRW11}.  Can we develop
similar algorithms for maximum 1-planar graphs?
\fi
1-planar graphs of higher connectivity are also worth exploring.
For $\delta=3,4$, our constructed graphs have connectivity $\delta$,
which is the best one can hope for.  But for $\delta\geq 5$ our
constructed graphs have a cutvertex.  Are there better matching-bounds
for 5-connected 1-planar graphs with minimum degree 5?  In contrast
to planar graphs, 5-connected 1-planar graphs do not necessarily
have a Hamiltonian path, and therefore not necessarily a {\em near-perfect
matching} of size $\lceil (n-1)/2 \rceil$ \cite{FHM+20}.
Do all 6-connected 1-planar graphs have a near-perfect matching?
How about all 7-connected ones?

%%%%%%%%%%%%%%%%%%%%%%%%%%%%%%%%%%%%%%%%%%%%%%%%%%%%%%%%%%%%%%%%%%%%%%%%
\bibliographystyle{plain}
\bibliography{journal,full,gd,papers}

\iffalse
%%%%%%%%%%%%%%%%%%%%%%%%%%%%%%%%%%%%%%%%%%%%%%%%%%%%%%%%%%%%%%%%%%%%%%%%
\begin{appendix}

\section{Appendix Stuff}

\end{appendix}
\fi

\end{document}